\documentclass[a4paper,fleqn]{cas-dc}

\usepackage[numbers]{natbib}

\usepackage{amsmath}
\usepackage{amsthm}
\usepackage{amsfonts}
\usepackage{amssymb}
\usepackage{algorithm}
\usepackage{algorithmic}

\theoremstyle{plain}
\newtheorem{theorem}{Theorem}

\newtheorem{lemma}{Lemma}
\newtheorem{proposition}{Proposition}
\newtheorem{corollary}{Corollary}

\theoremstyle{definition}
\newtheorem{definition}{Definition}

\newtheorem{remark}{Remark}

\def\tsc#1{\csdef{#1}{\textsc{\lowercase{#1}}\xspace}}
\tsc{WGM}
\tsc{QE}
\tsc{EP}
\tsc{PMS}
\tsc{BEC}
\tsc{DE}

\begin{document}
\let\WriteBookmarks\relax
\def\floatpagepagefraction{1}
\def\textpagefraction{.001}
\shorttitle{Distributed Phase Estimation Algorithm and Distributed Shor's Algorithm}
\shortauthors{L. Xiao and D. Qiu et~al.}

\title [mode = title]{Distributed Phase Estimation Algorithm and Distributed Shor's Algorithm}                      





                

                



                \author[1]{Ligang Xiao}[style=chinese]
\credit{Methodology, Writing - original draft}

\author[1]{Daowen Qiu}[style=chinese]
\cormark[1]
\ead{issqdw@mail.sysu.edu.cn}
\credit{Conceptualization of this study, Ideas and Writing, Review, Editing}

\author[2]{Le Luo}[style=chinese]
\credit{Review}

\author[3]{Paulo Mateus}
\credit{Review}

\affiliation[1]{organization={School of Computer Science and Engineering, Sun Yat-sen University},
                city={Guangzhou},
                postcode={510006}, 
                country={China}}
                

                

\affiliation[2]{organization={School of Physics and Astronomy, Sun Yat-sen University},
			city={Zhuhai},
			postcode={519082},
                country={China}}

\affiliation[3]{organization={Instituto de Telecomunica\c{c}\~{o}es, Departamento de Matem\'{a}tica, Instituto Superior T\'{e}cnico},
			addressline={Av. Rovisco Pais 1049-001},
			addresslinesep={},
			city={Lisbon},
                country={Portugal}}

\cortext[cor1]{Corresponding author}

\begin{abstract}
Shor's algorithm is one of the most significant quantum algorithms. Shor's algorithm can factor large integers with a certain success probability in polynomial time. However, Shor's algorithm requires an unbearable amount of qubits in the NISQ (Noisy Intermediate-scale Quantum) era. To reduce the resources required for Shor's algorithm, in this paper we first propose a new distributed phase estimation algorithm. Our distributed phase estimation algorithm does not require quantum communication and it reduces the number of qubits of a single node compared to the traditional phase estimation algorithm (non-iterative version). Then we apply our distributed phase estimation algorithm to form a distributed order-finding algorithm for Shor's algorithm. Compared with the traditional Shor's algorithm (non-iterative version), the maximum number of qubits required by a single node of our dristributed order-finding algorithm is reduced by  $(2-\dfrac{2}{k})L-\log_2k-O(1)$ when factoring an $L$-bit integer ($k$ is the number of compute nodes). The communication complexity of our distributed order-finding algorithm is $O(kL)$. 
\end{abstract}



\maketitle

\section{Introduction}\label{sec:introduction}

	Quantum computing is rapidly developing and has shown impressive advantages over classical computing in factoring larger integer \cite{shor1994algorithms}, solving linear system of equations \cite{harrow2009quantum}, simulating chemical molecular \cite{aspuru2005simulated} and other fields. However, in order to realize quantum algorithms in practice, medium or large scale general quantum computers are required. Currently it is still  difficult to implement such quantum computers. Therefore,  to advance the application of quantum algorithms in the NISQ era, we would consider to reduce the required qubits or other quantum resources for quantum computers.
	
	Distributed quantum computing is a computing method that combines distributed computing and quantum computing \cite{avron2021quantum,beals2013effcient,li2017application,yimsiriwattana2004distributed}. It aims to solve problems by utilizing multiple smaller quantum computers working together. Distributed quantum computing is usually used to reduce the resources required by each computer, including qubits, gate complexity, circuit depth and so on. Due to these potential benefits, distributed quantum computing has been studied significantly \cite{avron2021quantum,beals2013effcient,izumi2019quantum,le2018sublinear,li2017application,qiu2022distributed,tan2022distributed,yimsiriwattana2004distributed}.
For example, in 2013, Beals et al. proposed an algorithm for parallel addressing quantum memory \cite{beals2013effcient}.	In 2018, Le Gall et al. studied quantum algorithms in the quantum CONGEST model \cite{le2018sublinear}.  In 2022, Qiu et al. proposed a distributed Grover's algorithm \cite{qiu2022distributed}, and Tan, Xiao and Qiu et al. proposed a distributed quantum algorithm for Simon's problem \cite{tan2022distributed}. 
These distributed quantum algorithms can reduce quantum resources to some extent.
	
	Shor's algorithm \cite{shor1994algorithms}  is  one of the most significant algorithms in quantum computing. It can factor large integers with a certain probability of success and costs polynomial time. Since the best known classical algorithm for factoring large numbers is subexponential but superpolynomial, Shor's algorithm demonstrates quantum advantages. Shor's algorithm can be applied to break RSA encryption which has been widely used in public key cryptography system. Shor's algorithm can be implemented in two ways: one needs to measure multiple qubits at the end (we call it \textit{non-iterative Shor's algorithm}, e.g. \cite{kaye2006introduction,nielsen2000quantum,shor1994algorithms}), and the other alternately performs unitary operators and measurements, and only  one qubit is measured at a time (we call it \textit{iterative Shor's algorithm}, e.g. \cite{beauregard2003circuit,haner2017factoring,parker2000efficient}). The iterative Shor's algorithm has only one control qubit and it requires $2L+O(1)$ qubits when factoring an $L$-bit integer \cite{beauregard2003circuit,haner2017factoring,parker2000efficient}. The non-iterative Shor's algorithm has $2L+O(1)$ control qubits and thus it requires $4L+O(1)$ qubits \cite{kaye2006introduction,nielsen2000quantum,shor1994algorithms}.
	
	Shor's algorithm (proposed in 1994) contains the idea of phase estimation, but the phase estimation algorithm was formally proposed in 1995 \cite{kitaev1995quantum}. Similar to Shor's algorithm, we divide the phase estimation algorithm into \textit{non-iterative phase estimation algorithm} (e.g. \cite{kaye2006introduction,nielsen2000quantum}) and \textit{iterative phase estimation algorithm} (e.g. \cite{kitaev1995quantum}). The iterative phase estimation algorithm requires one control qubits and the non-iterative phase estimation algorithm requires multiple control qubits.
	
	However, the application of phase estimation algorithm or Shor's algorithm requires a large number of qubits \cite{craig2021how}. Therefore, it is very necessary to reduce the resources required in these algorithms by designing new methods, such as utilizing distributed quantum computing. 
	
	In 2017, Li and Qiu et al. \cite{li2017application} proposed a distributed phase estimation algorithm. Actually, their method can reduce the number of qubits or gate complexity required by a single node. However, their algorithm cannot guarantee that the deviation  of the final result from real result is within a given range. In 2020, Neumann et al. proposed a distributed phase estimation algorithm by implementing non-local gates \cite{neumann2020imperfect}. In fact, it is a general distributed approach.
	
	The first distributed Shor's algorithm was proposed by Yimsiriwattana et al. \cite{yimsiriwattana2004distributed} in 2004. In a way, this is also a universal distributed method. Their method first divides the quantum circuit into several parts directly, and then realizes non-local quantum gates by means of quantum communication. The maximum number of qubits required by a single node of their distributed algorithm is $L+O(1)$ when factoring an $L$-bit integer.  Its communication complexity is $O(L^2)$. 
	
	Recently, Xiao and Qiu et al. \cite{gang2023distributed} proposed a distributed Shor's algorithm. Their algorithm uses two nodes to cooperate to complete the key step in Shor's algorithm, that is, to estimate some $\dfrac{s}{r}$, where $r$ is the ``order'' and $s\in\{0,\cdots,r-1\}$. Compared with the traditional Shor's algorithm, their algorithm reduces the required qubits (nearly $L/2$ qubits are reduced) of a single node. In addition, the communication complexity of their distributed Shor's algorithm is $O(L)$.

	In this paper, we first propose a new distributed phase estimation algorithm. Our distributed phase estimation algorithm utilizes multiple nodes to estimate bits at different positions of the phase and employs classical post-processing to adjust the deviation of the final result. Our distributed phase estimation algorithm does not require quantum communication, and each node requires $\dfrac{n}{k}+\log_2k+O(1)$ control qubits ($k$ is the number of nodes) when estimating the first $n$ bits of phase. Compared with the non-iterative phase estimation algorithm, the maximum number of qubits required by a single node of our distributed algorithm is reduced by $(1-\dfrac{1}{k})n-\log_2k-O(1)$.
	
	Afterwards, we apply the above distributed phase estimation algorithm to form a distributed Shor's algorithm (more specifically, to form a order-finding algorithm). The maximum number of qubits required by a single node in our distributed order-finding algorithm is $(2+\dfrac{1}{k})L+\log_2k+O(1)$ ($k$ is the number of nodes) when factoring an $L$-bit integer, and its communication complexity is $O(kL)$, which is better than Yimsiriwattana's algorithm \cite{yimsiriwattana2004distributed}. 
	
	
	The remainder of the paper is organized as follows. First, in Section 2, we review a number of quantum algorithms related to phase estimation algorithm and Shor's algorithm.  Then in Section 3, we present our distributed phase estimation algorithm. After that, in Section 4, we present our main result---distributed order-finding algorithm, and subsequently, in Section 5, we analyze the complexity of our distributed algorithms and compare them with other related algorithms. Finally in Section 6, we summarize the main results and mention potential problems for further study.

\section{Preliminaries}\label{sec:preliminaries}

In this section, we would review  quantum Fourier transform, phase estimation algorithm, order-finding algorithm, and other relevant concepts that will be used in the paper. It is assumed that the readers have a familiarity with linear algebra and basic notation in quantum computing. In the interest of readability, we  review some basic concepts concerning quantum computing in Appendix 1, and  for further details, we can refer to \cite{nielsen2000quantum}.

\subsection{Quantum Fourier transform}

	Quantum Fourier transform is a unitary operator that acts on the standard basis states as follows:
\begin{equation}
QFT |j\rangle=\frac{1}{\sqrt{2^n}}\sum_{k=0}^{2^n-1}e^{2\pi ijk/2^n}|k\rangle\text{,}
\end{equation}
for $j=0,1,\cdots,2^n-1$. Therefore, the inverse quantum Fourier transform acts as follows:
\begin{equation}\label{inverse_QFT}
QFT^{-1} \frac{1}{\sqrt{2^n}}\sum_{k=0}^{2^n-1}e^{2\pi ijk/2^n}|k\rangle=|j\rangle\text{,}
\end{equation}
for $j=0,1,\cdots,2^n-1$.
The quantum Fourier transform and its inverse can be implemented by using $O(n^2)$ elementary gates (i.e., $O(n^2)$ single-qubit and two-qubit gates) \cite{nielsen2000quantum,shor1994algorithms}.

\subsection{Phase estimation algorithm}

	Phase estimation algorithm is a practical application of quantum Fourier transform. Consider a quantum state $|u\rangle$ and a unitary operator $U$ such that 
\begin{equation}\label{eq:Uu}
U|u\rangle=e^{2\pi i\omega}|u\rangle
\end{equation}
 for some real number $\omega\in[0,1)$. If we can implement controlled operation $C_m(U)$ satisfying that
\begin{equation}\label{CmU}
C_m(U)|j \rangle|u\rangle=|j\rangle U^j|u\rangle
\end{equation}
for any positive integer $m$ and $m$-bit string $j$, where the first register is control qubits, then we can apply the phase estimation algorithm to estimate $\omega$ (see Algorithm \ref{alg:PE}).

For the sake of convenience, we give a number of notations in the following definition. 

\begin{definition}
For any real number $\omega$, suppose its binary representation is $\omega=a_{1}a_{2}\cdots a_l.b_1b_2\cdots$. Denote $|PE_{t,\omega}\rangle, \omega_{\{i,j\}}, \omega_{[i,j]} $, $d_t(x,y)$ , $len(x)$, $ADD(x,b)$ respectively as follows:
\begin{itemize}
\item $|PE_{t,\omega}\rangle$: for any positive integer $t$, define\\
$|PE_{t,\omega}\rangle=QFT^{-1}\dfrac{1}{\sqrt{2^t}}\sum\limits_{j=0}^{2^t-1}e^{2\pi ij\omega}|j\rangle $.

\item $\omega_{\{i,j\}}$: for any integer $i,j$ with $1\leq i\leq j$, define\\
 $\omega_{\{i,j\}}=b_i b_{i+1}\cdots b_j$.
\item $\omega_{\{i,+\infty\}}$: for any positive integer $i$, define\\
 $\omega_{\{i,+\infty\}}=0.b_ib_{i+1}\cdots $.
\item $\omega_{[i,j]}$: for any integer $i,j$ with $1\leq i\leq j\leq l$, define\\
$\omega_{[i,j]}=a_{i} a_{i+1}\cdots a_j$.

\item $d_t(x,y)$: for any $x,y\in\{0,1\}^t$, define\\
 $d_t(x,y)=\min\left(|x-y|, 2^t-|x-y|\right)$.
\item $len(x)$: the length of string $x$.
\item $ADD(x,b)$: for any bit string $x$ and integer $b$, $ADD(x,b)$ is a bit string $y$ of length $len(x)$, with $y=(x+b)\mod 2^{len(x)}$.
\end{itemize}
\end{definition}

\begin{remark}
In this paper, when performing operations or comparisons on bit strings, we consider them as their corresponding binary numbers. The definitions of $d_t(x,y)$ and $ADD(x,b)$ follow this principle. 
\end{remark}

\begin{algorithm}[h]
\caption{Phase estimation algorithm}
\label{alg:PE}
\textbf{Input}: A positive integer $n$ (it means that we want to estimate the first $n$ bits of $\omega$) and the success probability $1-\epsilon$ ($\epsilon\in(0,1)$).

\textbf{Output}: A $t$-bit string $\widetilde{\omega}$ such that $d_n(\widetilde{\omega}_{[1,n]},\omega_{\{1,n\}})\leq 1$.
\textbf{Procedure}:
\begin{algorithmic}[1]
\STATE Create initial state $|0\rangle|u\rangle$:\\
The first register is $t$-qubit.
\STATE Apply $H^{\otimes t}$ to the first register:\\
  \quad$H^{\otimes t}|0\rangle|u\rangle=\dfrac{1}{\sqrt{2^t}}\sum\limits_{j=0}^{2^t-1}|j\rangle|u\rangle$.
\STATE Apply $C_t(U)$:\\
  \quad$C_t(U)\dfrac{1}{\sqrt{2^t}}\sum\limits_{j=0}^{2^t-1}|j\rangle|u\rangle=\dfrac{1}{\sqrt{2^t}}\sum\limits_{j=0}^{2^t-1}|j\rangle U^j|u\rangle=\dfrac{1}{\sqrt{2^t}}\sum\limits_{j=0}^{2^t-1}|j\rangle e^{2\pi ij\omega}|u\rangle$.
\STATE Apply $QFT^{-1}$:\\
 \quad$QFT^{-1}\dfrac{1}{\sqrt{2^t}}\sum\limits_{j=0}^{2^t-1}e^{2\pi ij\omega}|j\rangle |u\rangle=|PE_{t,\omega}\rangle|u\rangle$.
\STATE Measure the first register:\\
 \quad obtain a $t$-bit string $\widetilde{\omega}$.

\end{algorithmic}
\end{algorithm}

\begin{remark}\label{remark:PEVariant}
Let $x$ be a natural number. By Equation (\ref{eq:Uu}), we have $U^{2^{x-1}}|u\rangle=e^{2\pi i(2^{x-1}\omega)}|u\rangle=e^{2\pi i\omega_{\{x,+\infty\}}}|u\rangle$. Thus, to estimate $\omega_{\{x,+\infty\}}$, we can apply the phase estimation algorithm similarly and change $C_t(U)$ in step 3 to $C_t(U^{2^{x-1}})$ accordingly \cite{li2017application}.
\end{remark}

$d_t(\cdot,\cdot)$ has the following properties.

\begin{lemma}[See \cite{gang2023distributed}]\label{d_t}
Let $t$ be a positive integer and let $x,y$ be any two $t$-bit strings. It holds that: \\
{\rm (I)} Let $B=\{b\in\{-(2^t-1),\cdots,2^{t}-1\}: ADD(x,b)=y\}$. Then $d_t(x,y)=\min_{b\in B}|b|$.\\
{\rm (II)} $d_t(\cdot,\cdot)$ is a distance on $\{0,1\}^t$.\\
{\rm (III)} Let $t_0<t$ be an positive integer. If $d_t(x,y)<2^{t-t_0}$, then
\begin{equation}\label{d_t 1}
d_{t_0}(x_{[1,t_0]},y_{[1,t_0]})\leq 1.
\end{equation}
\end{lemma}

	The goal of  phase estimation algorithm is to estimate $\omega$, which can be more accurately described by the following propositions.

\begin{proposition}[See \cite{nielsen2000quantum}]\label{phase_estimation_result}
	In {\rm Algorithm \ref{alg:PE}}, for any $\epsilon>0$ and any positive integer $n$, if $t=n+\lceil\log_2(2+\dfrac{1}{2\epsilon})\rceil$, then the probability of $d_t(\widetilde{\omega},\omega_{\{1,t\}})<2^{t-n}$ is at least $1-\epsilon$. 
\end{proposition}

\begin{proposition}\label{phase_estimation2}
	In {\rm Algorithm \ref{alg:PE}}, for any $\epsilon>0$ and any positive integer $n$, if $t=n+\lceil\log_2(2+\dfrac{1}{2\epsilon})\rceil$, then the probability of $d_n(\widetilde{\omega}_{[1,n]},\omega_{\{1,n\}})\leq 1$ is at least $1-\epsilon$. 
\end{proposition}
\begin{proof}
Immediate from Proposition \ref{phase_estimation_result} and Lemma \ref{d_t}.
\end{proof}
\begin{remark}
According to these two propositions, we know that the phase estimation algorithm can not get an estimate $\widetilde{\omega}$ that is arbitrarily close to $\omega$. It only satisfies $d_n(\widetilde{\omega}_{[1,n]},\omega_{\{1,n\}})\leq 1$, not $|\widetilde{\omega}_{[1,n]}-\omega_{\{1,n\}}|\leq 1$. However, if $1\leq \omega_{\{1,n\}}< 2^n-1$, we can conclude $|\widetilde{\omega}_{[1,n]}-\omega_{\{1,n\}}|\leq 1$ from $d_n(\widetilde{\omega}_{[1,n]},\omega_{\{1,n\}})\leq 1$.
\end{remark}

	By means of  using mathematical language to describe Proposition \ref{phase_estimation2}, the following corollary can be obtained.

\begin{corollary}\label{MeasurePE}
Let $n$ be a positive integer and let $\omega\in[0,1),\epsilon\in(0,1),t=n+\lceil\log_2(2+\dfrac{1}{2\epsilon})\rceil$. Denote $D=\{x\in\{0,1\}^n: d_n(x,\omega_{\{1,n\}})\leq 1\}$ and $P_D=\sum_{a\in D}|a\rangle\langle a|$, then
\begin{equation}
\left\|P_D|PE_{t,\omega}\rangle \right\|^2\geq 1-\epsilon.
\end{equation}
\end{corollary}


\subsection{Order-finding algorithm}

 Order-finding algorithm is the key subroutine in Shor's algorithm. Given an $L$-bit integer $N$ and a positive integer $a$ with $gcd(a,N)=1$, the goal of order-finding algorithm is to determine the order $r$ of $a$ modulo $N$, where the order $r$ is defined as the smallest integer $r$ such that $a^r\equiv 1(\bmod\ N)$. An important unitary operator  $M_a$ in order-finding algorithm is defined as
\begin{equation}
M_a|x \rangle=|ax\ \bmod\ N\rangle \text{.}
\end{equation}
Define
\begin{equation}
|u_s\rangle=\dfrac{1}{\sqrt{r}}\sum\limits_{k=0}^{r-1}e^{-2\pi i\frac{s}{r}k}|a^k \bmod\ N\rangle,
\end{equation}
$s=0,1,\cdots,r-1$. It satisfies
\begin{align}
M_a|u_s\rangle=e^{2\pi i\frac{s}{r}} |u_s\rangle,\\
\dfrac{1}{\sqrt{r}}\sum\limits_{s=0}^{r-1}|u_s\rangle=|1\rangle,
\end{align} and
\begin{equation}\label{us_orthonormal}
\langle u_s|u_{s'}\rangle=\delta_{s,s'}=
\begin{cases} 0 &\text{if $s\not= s'$},\\
1 &\text{if $s=s'$}.
\end{cases}
\end{equation}

Algorithm \ref{alg:OFA}  \cite{nielsen2000quantum} and Figure \ref{fig:order finding algorithm} show the procedure of order-finding algorithm.

\begin{algorithm}[h!]
\caption{Order-finding algorithm}
\label{alg:OFA}
\textbf{Input}: Positive integers $N$ and $a$ with $gcd(N,a)=1$.\\
\textbf{Output}: The order $r$ of $a$ modulo $N$.\\
\textbf{Procedure}:
\begin{algorithmic}[1]
\STATE Create initial state $|0\rangle|1\rangle$:\\
\quad The first register has $t=2L+1+\lceil\log_2(2+\dfrac{1}{2\epsilon})\rceil$ qubits and the second register has $L$ qubits.
\STATE Apply $H^{\otimes t}$ to the first register:\\
 \quad $H^{\otimes t}|0\rangle|1\rangle=\dfrac{1}{\sqrt{2^t}}\sum\limits_{j=0}^{2^t-1}|j\rangle|1\rangle$.
\STATE Apply $C_t(M_a)$:\\
  \quad$C_t(M_a)\dfrac{1}{\sqrt{2^t}}\sum\limits_{j=0}^{2^t-1}|j\rangle|1\rangle=$\\$\dfrac{1}{\sqrt{2^t}}\sum\limits_{j=0}^{2^t-1}|j\rangle M^j (\dfrac{1}{\sqrt{r}}\sum\limits_{s=0}^{r-1}|u_s\rangle)=$\\$\dfrac{1}{\sqrt{r2^t}}\sum\limits_{s=0}^{r-1}\sum\limits_{j=0}^{2^t-1}|j\rangle e^{2\pi ij\frac{s}{r}}|u_s\rangle$.
\STATE Apply $QFT^{-1}$:\\
 \quad $QFT^{-1}\dfrac{1}{\sqrt{r2^t}}\sum\limits_{s=0}^{r-1}\sum\limits_{j=0}^{2^t-1}|j\rangle e^{2\pi ij\frac{s}{r}}|u_s\rangle=$\\$\dfrac{1}{\sqrt{r}}\sum\limits_{s=0}^{r-1}|PE_{t,s/r}\rangle|u_s\rangle$
\STATE Measure the first register:\\
\quad obtain a $t$-bit string $m$ (it can be an estimation of $\dfrac{s}{r}$ for some $s$).
\STATE Apply continued fractions algorithm:\\
\quad obtain $r$.
\end{algorithmic}
\end{algorithm}

\begin{figure}[h!]
	\centering
	\includegraphics[width=0.4\textwidth]{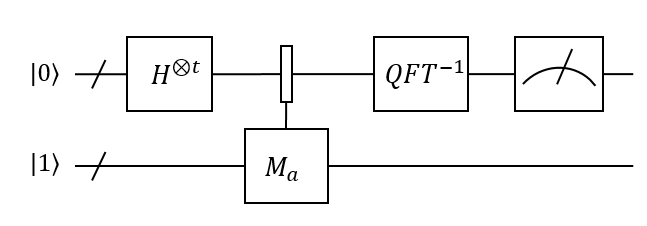}
	\caption{\label{fig:order finding algorithm} Circuit for order-finding algorithm}
\end{figure}

	The function of the quantum part of Algorithm \ref{alg:OFA} (steps 1 to 5) can be described by the following proposition.
	
\begin{proposition}[See \cite{nielsen2000quantum}]\label{OrderFindingResult}
In Algorithm \ref{alg:OFA}, for any fixed $s\in \{0,1,\cdots,r-1\}$, the probability of
$$\left|\dfrac{m}{2^{len(m)}}-\dfrac{s}{r}\right|<2^{-(2L+1)}$$  is at least $\dfrac{1-\epsilon}{r}$. 
\end{proposition}

	By Proposition \ref{OrderFindingResult}, we can see that the goal of the quantum part of  Algorithm \ref{alg:OFA} (steps 1 to 5) is to obtain an estimation of a random $\dfrac{s}{r}$ where $s\in\{0,1,\cdots,r-1\}$ (i.e. $\left|\dfrac{m}{2^t}-\dfrac{s}{r}\right|\leq 2^{-(2L+1)}$).  It is worth mentioning that the deviation range $2^{-(2L+1)}$ is a prerequisite to ensure the correctness of step 6 in Algorithm \ref{alg:OFA}.

\section{Distributed phase estimation algorithm}

	In 2017, Li and Qiu et al. \cite{li2017application} proposed a distributed phase estimation algorithm, which employs the technique mentioned in Remark \ref{remark:PEVariant}. However,  the deviation  in their algorithm may not be  within a given range. In this section, we propose a new distributed phase estimation algorithm. By combining some classical post-processing strategies, we ensure the correctness of our distributed algorithm. Suppose $U,|u\rangle,\omega$ satisfy Equation (\ref{eq:Uu}) and we estimate the first $n$ bits of $\omega$.  The idea of our algorithm is as follows:

	Let integers $l_1,l_2,\cdots,l_{k+1}$ satisfy 
\begin{equation}\label{eq:li}
0=l_1<l_2\cdots <l_{k+1}=n-2.
\end{equation} 
We use $k$ computing nodes (denoted as $A_1,\cdots,A_k$) to estimate the bits of different parts of $\omega$ respectively, where node $A_i$ estimates $\omega_{\{l_i,l_{i+1}+2\}}$, $i=1,2,\cdots,k+1$ (shown in Fig. \ref{fig:DPEEstimation}). We can do this by employing the technique mentioned in Remark \ref{remark:PEVariant}. It can be seen that the corresponding positions of the first three bits of $A_{i+1}$'s estimation and the last three bits of $A_{i+1}$'s estimation are overlapping with each other. So we can use the bits with overlapped positions to correct the estimation results and finally combine all estimates.The process of using phase estimation in a distributed manner is demonstrated in Algorithm \ref{alg:DPE}. The steps of correction and combination are shown in Algorithm \ref{alg:CorrectAndCombine}.

\begin{figure}[h!]
	\centering
	\includegraphics[width=0.45\textwidth]{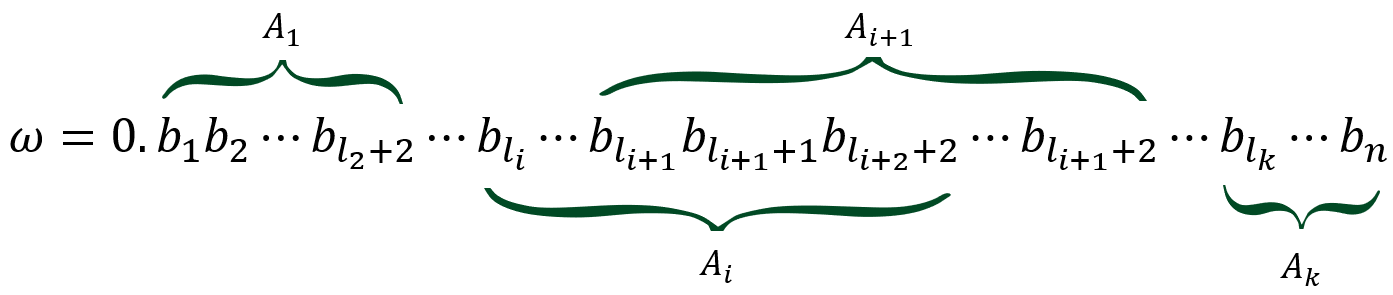}
	\caption{\label{fig:DPEEstimation} The positions of bits estimated by each node}
\end{figure}

\begin{algorithm}[h!]
\caption{Distributed phase estimation algorithm ($k$ nodes)}
\label{alg:DPE}

\textbf{Input}: A positive integer $n$ (it means that we want to estimate the first $n$ bits of $\omega$) and the success probability $1-\epsilon$ ($\epsilon\in(0,1)$).\\
\textbf{Output}: Output an $n$-bit string $m$ such that $d_t(m,\omega_{\{1,n\}})\leq 1$ with success probability at least $1-\epsilon$.\\
\textbf{Procedure}:\\
Node $A_1,A_2,\cdots,A_k$ perform the following operations in parallel.

\quad\textbf{ Node $A_r$ excute ($r=1,2,\cdots,k$)}:
\begin{algorithmic}[1]

\STATE Create initial state $|0\rangle_{R_r}|u\rangle$:\\
Register $R_r$ is $t_r$-qubit, where $t_r=l_{r+1}+3-l_r+\lceil\log_2(2+\dfrac{k}{2\epsilon})\rceil$.
\STATE Apply $H^{\otimes t_r}$ to the first register:\\
  \quad$H^{\otimes t_r}|0\rangle_{R_r}|u\rangle=\dfrac{1}{\sqrt{2^{t_r}}}\sum\limits_{j=0}^{2^{t_r}-1}|j\rangle_{R_r}|u\rangle$.
\STATE Apply $C_{t_r}(U^{2^{l_r-1}})$:\\
\quad$C_{t_r}(U^{2^{l_r-1}})\dfrac{1}{\sqrt{2^{t_r}}}\sum\limits_{j=0}^{2^{t_r}-1}|j\rangle_{R_r}|u\rangle=$\\
$\dfrac{1}{\sqrt{2^{t_r}}}\sum\limits_{j=0}^{2^{t_r}-1}|j\rangle_{R_r} (U^{2^{l_r-1}})^j|u\rangle=$\\
$\dfrac{1}{\sqrt{2^{t_r}}}\sum\limits_{j=0}^{2^{t_r}-1}|j\rangle_{R_r} e^{2\pi ij(2^{l_r-1}\omega)}|u\rangle=$\\
$\dfrac{1}{\sqrt{2^{t_r}}}\sum\limits_{j=0}^{2^{t_r}-1}|j\rangle_{R_r} e^{2\pi ij\omega_{\{l_r,+\infty\}}}|u\rangle$.
\STATE Apply $QFT^{-1}$:\\
 $QFT^{-1}\dfrac{1}{\sqrt{2^{t_r}}}\sum\limits_{j=0}^{2^{t_r}-1}e^{2\pi ij\omega_{\{l_r,+\infty\}}} |j\rangle_{R_r} |u\rangle=$\\$|PE_{t_r,\omega_{\{l_r,+\infty\}}}\rangle_{R_r}|u\rangle$.
 
 \STATE Measure the first $l_{r+1}+3-l_r$ bits of its register $R_r$:\\
 denote the measuring result of $A_r$ as $m_r$.

\textbf{Any node executes}:
\quad 
\STATE $m\leftarrow CorrectAndCombine(m_1,\cdots,m_k)$:\\
\quad $m$ is an $n$-bit string.
\STATE Return $m$.
\end{algorithmic}
\end{algorithm}

\begin{algorithm}[h!]
\caption{CorrectAndCombine}
\label{alg:CorrectAndCombine}
\textbf{Input}: $k$ bit strings $m_1,\cdots,m_k$, where $len(m_i)\geq 3$, $i=1,2,\cdots,k$ ($k$ can be any positive integer).\\
\textbf{Procedure}:
\begin{algorithmic}[1]
\STATE Set $m_k'=m_k$.
\FOR{$r=k-1$ to $1$}
	\STATE Choose $CorrectionNum_r\in\{\pm 2,\pm 1,0\}$ such that\\
	$ADD\left((m_r)_{[l_{r+1},l_{r+1}+2]},CorrectionNum_r\right)=$$(m_{r+1}')_{[1,3]}$
	\STATE $prefix_r\leftarrow ADD(m_r,CorrectionNum_r)$
	\STATE $m_r'\leftarrow prefix_r\circ (m_{r+1}')_{[4,len(m_{r+1}')]}$ (``$\circ$" represents catenation)
\ENDFOR

\STATE Return $m_1'$.
\end{algorithmic}
\end{algorithm}

The function of our distributed phase estimation algorithm is the same as that of the traditional phase estimation algorithm. It can be described by the following theorem.

\begin{theorem}\label{CorrectnessForDPE}
In Algorithm \ref{alg:DPE}, the probability of 
\begin{equation}
d_n(m,\omega_{\{1,n\}})\leq 1
\end{equation}
is at least $1-\epsilon$. 
\end{theorem}
\begin{proof}
See Appendix.
\end{proof}

\section{Distributed order-finding algorithm} \label{sec:distributed order finding algorithm}
	Recently, Xiao and Qiu et al. proposed a distributed Shor's algorithm that requires two compute nodes \cite{gang2023distributed}. Compared with the traditional Shor's algorithm, their distributed Shor's algorithm can reduce nearly $\dfrac{L}{2}$ qubits and reduce circuit depth to some extent for each node when factoring an $L$-bit composite number. In addition, their communication complexity is $O(L)$, which is better than that of the distributed Shor's algorithm in \cite{yimsiriwattana2004distributed} (its communication complexity is $O(L^2)$ ).
	
	In this section, by applying distributed phase estimation algorithm (Algorithm \ref{alg:DPE}), we propose a new multi-node distributed Shor's algorithm. Compared to the distributed Shor's algorithm proposed by Xiao and Qiu et al. (denoted as Xiao's algorithm, for simplicity), our distributed Shor's algorithm has the following advantages:
\begin{itemize}
\item Xiao's algorithm only utilizes two nodes, but our algorithm utilizes multiple nodes.
\item When factoring an $L$-bit integer, in Xiao's algorithm, according to their proof, it can be inferred that the last node must estimate more than $L+2$ bits. However, in our algorithm, each node only needs to ensure to estimate more than $2$ bits. This is because the idea of distributed phase estimation hidden in Xiao's algorithm is limited, while our distributed phase estimation algorithm is universal.
\end{itemize}

	Next, we introduce our distributed order-finding algorithm. Firstly, it should be noted that we cannot directly apply Algorithm \ref{alg:DPE} to the order-finding algorithm, since if we do this, the estimated bits for each node do not correspond to a same $\dfrac{s}{r}$, where $r$ is the ``order'' and $s\in\{0,1,\cdots,r-1\}$. To solve this problem, we need to employ quantum communication. Afterwards, Algorithm \ref{alg:DPE} can be applied. When factoring an $L$-bit integer, in order to ensure that the final estimation result $m$ satisfy $\left|\dfrac{m}{2^{len(m)}}-\dfrac{s}{r}\right|<2^{-(2L+1)}$, we need to estimate $2L+2$ bits. This is because if a bit string $x$ of length $2L+2$ satisfy $\left|x-\left(\dfrac{s}{r}\right)_{\{1,2L+2\}}\right|\leq 1$, then
\begin{equation}
\left|\dfrac{x}{2^{len(x)}}-\dfrac{s}{r}\right|\leq 2^{-(2L+1)}
\end{equation}
holds. Let integers $l_1,\cdots,l_{k+1}$ satisfy $1=l_1<\cdots<l_{k+1}=2L$. Our distributed order-finding algorithm is shown in Algorithm \ref{alg:DOFA}.  Figure \ref{fig:DOFA} shows the quantum circuit of Algorithm \ref{alg:DOFA}.

\begin{algorithm}[h!]
\caption{Distributed order-finding algorithm ($k$ nodes)}
\label{alg:DOFA}
\textbf{Input}: Positive integers $N$ and $a$ with $gcd(N,a)=1$ and $a<N$.\\
\textbf{Output}: The order $r$ of $a$ modulo $N$.\\
\textbf{Procedure}:
\begin{algorithmic}[1]
\STATE Node $A_1$ creates initial state $|0\rangle_{R_1}|1\rangle_C$. Node $A_2,\cdots,A_k$ create initial states $|0\rangle_{R_2},\cdots,|0\rangle_{R_k}$ respectively: \\
 \quad Here register $C$ is $L$-qubit and register $R_j$ is $t_j$-qubit, where $t_j=l_{j+1}+3-l_j+\lceil\log_2(2+\dfrac{k}{2\epsilon})\rceil$, $j=1,\cdots,k$.
 
\STATE Set $u=1$.

\textbf{Node  $A_u$ executes}:
\STATE Apply $H^{\otimes t_u}$ to register $R_u$:
\begin{align}
&\hspace{-2em}\dfrac{1}{\sqrt{r}}\sum\limits_{s=0}^{r-1}\left(\mathop{\otimes}\limits_{j=1}^{u-1}|PE_{t_j,\left(\frac{s}{r}\right)_{\{l_j,+\infty\}}}\rangle_{R_j}\right)  H^{\otimes t_u}|0\rangle_{R_u}|u_s\rangle_C\nonumber\\
 &\left(\mathop{\otimes}\limits_{j=u+1}^k|0\rangle_{R_j}\right)\nonumber
\end{align}

\STATE Apply $C_{t_u}(M_a^{2^{l_u-1}})$ to registers $R_u$ and $C$:
\begin{align}
 &\hspace{-2.7em}\dfrac{1}{\sqrt{r}}\sum\limits_{s=0}^{r-1}\left(\mathop{\otimes}\limits_{j=1}^{u-1}|PE_{t_j,\left(\frac{s}{r}\right)_{\{l_j,+\infty\}}}\rangle_{R_{j}}\right) \nonumber\\
  &\hspace{-2em}\left(\dfrac{1}{\sqrt{2^{t_u}}}\sum\limits_{j=0}^{2^{t_u}-1}e^{2\pi ij\left(\frac{s}{r}\right)_{\{l_u,+\infty\}}}|j\rangle_{R_u} |u_s\rangle_C\right) \left(\mathop{\otimes}\limits_{j=u+1}^k|0\rangle_{R_j}\right)\nonumber
\end{align}

\STATE Apply $QFT^{-1}$to register $R_u$:
$$\dfrac{1}{\sqrt{r}}\sum\limits_{s=0}^{r-1}\left(\mathop{\otimes}\limits_{j=1}^{u}|PE_{t_j,\left(\frac{s}{r}\right)_{\{l_j,+\infty\}}}\rangle_{R_{j}}\right) |u_s\rangle_C \left(\mathop{\otimes}\limits_{j=u+1}^k|0\rangle_{R_j}\right)$$

\STATE If $u<k$, then teleport the qubits of register $C$ to node $A_{u+1}$:
$$\hspace{-1.6em}\dfrac{1}{\sqrt{r}}\sum\limits_{s=0}^{r-1}\left(\mathop{\otimes}\limits_{j=1}^{u}|PE_{t_j,\left(\frac{s}{r}\right)_{\{l_j,+\infty\}}}\rangle_{R_{j}}\right) |0\rangle_{R_{u+1}}|u_s\rangle_C \left(\mathop{\otimes}\limits_{j=u+2}^k|0\rangle_{R_j}\right)$$

\STATE If $u<k$, then set $u\leftarrow u+1$ and go to step 3.

\textbf{Finally}:
\STATE Node $A_j$ measures the first $l_{j+1}+3-l_j$ bits of its register $R_j$, $j=1,\cdots,k$:\\
\quad denote the measurement result of $A_j$ as $m_j$ , $j=1,\cdots,k$.
\STATE $m\leftarrow CorrectAndCombine(m_1,\cdots,m_k)$:\\
\quad $m$ is a $(2L+2)$-bit string.
\STATE Apply continued fractions algorithm:\\
\quad obtain $r$.
\end{algorithmic}
\end{algorithm}

\begin{figure}[h!]
	\centering
	\includegraphics[width=0.48\textwidth]{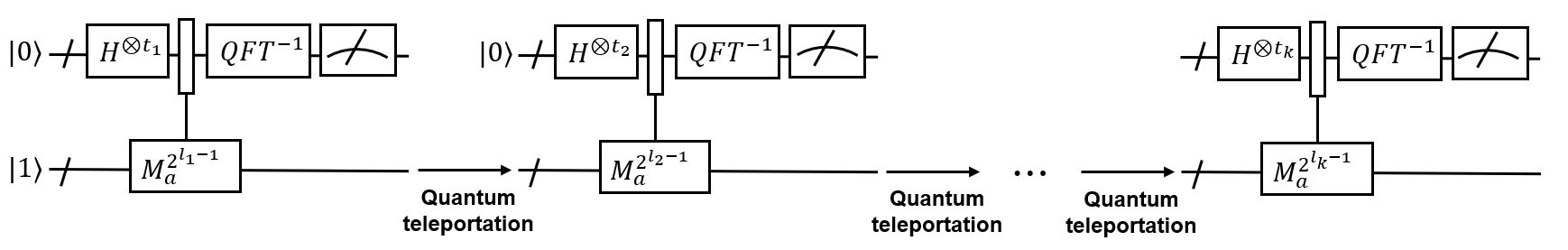}
	\caption{\label{fig:DOFA} Circuit for distributed order-finding algorithm}
\end{figure}

The function of this distributed phase estimation algorithm is the same as that of the order-finding algorithm (Algorithm \ref{alg:OFA} ). It can be described by the following theorem (it is almost the same as Proposition \ref{OrderFindingResult}), which indicates that Algorithm \ref{alg:DOFA} is correct.

\begin{theorem}\label{CorrectnessForDOFA}
In Algorithm \ref{alg:DOFA}, for any fixed $s\in \{0,1,\cdots,r-1\}$, the probability of
$$\left|\dfrac{m}{2^{len(m)}}-\dfrac{s}{r}\right|<2^{-(2L+1)}$$  is at least $\dfrac{1-\epsilon}{r}$. 
\end{theorem}
\begin{proof}
The proof is very similar to the proof of the correctness of applying the phase estimation algorithm to the order-finding algorithm, so we omit it.
\end{proof}

\section{Complexity analysis} \label{sec:complexity analysis}

In this section, we analyze the complexity of our distributed phase estimation algorithm and distributed Shor's algorithm, and compare them with other related algorithms.

\subsection{Complexity of distributed phase estimation algorithm}

	In phase estimation algorithm, since the specific structure of $U$ is unknown, we can not directly give its circuit depth or gate complexity. However, the main operator of phase estimation algorithm is $C_t(U)$  (or similar operator). According to the Figure 1 in \cite{gang2023distributed}, we know that the operator $C_t(U)$ can be implemented by $t$ controlled operators in the form of controlled-$U^x$, $x=1,2,3,4,\cdots$. Therefore, we take the number of  controlled-$U^x$ gates as a metric in the complexity analysis, since it has certain correlations with circuit depth and gate complexity.
	
	In Algorithm \ref{alg:DPE}, we choose appropriate values for $l_2,\cdots,l_k$. We make $\left|l_i-(i-1)\cdot\dfrac{n}{k}\right|<1$, $i=2,\cdots,k$. The complexity of Algorithm \ref{alg:DPE} and other related algorithms are shown in Table \ref{table:DPE}. In the general distributed method (\cite{neumann2020imperfect}) of the Table \ref{table:DPE}, the circuit of each node is not unique. Hence the number of controlled-$U^x$ gates per node is uncertain. Therefore, the corresponding position in the table is represented by ``undefined".

\begin{table}[width=1\linewidth,cols=4,pos=h]
\caption{Complexity of Algorithm \ref{alg:DPE} and other related algorithms} \label{table:DPE}
\begin{tabular*}{\tblwidth}{ p{1.8cm}| p{1.5cm}| p{1.5cm}| p{2cm}}
\bottomrule
Algorithms & qubits (per node) & quantum communication complexity& number of controlled-$U^x$ (per node)\\
\hline
Algorithm \ref{alg:DPE} & $\dfrac{n}{k}+\log_2k+B_u+O(1)$ & $0$ & $\dfrac{n}{k}+\log_2k+O(1)$ \\
\hline
traditional non-iterative PEA & $n+B_u+O(1)$ & 0 & $n+O(1)$ \\\hline
General distributed method (\cite{neumann2020imperfect}) & $
\dfrac{n+B_u}{k}+O(1)$ & $\Omega(n)$ & undefined \\
\toprule
\end{tabular*}
\end{table}


\begin{remark}
In phase estimation algorithm, we have $U|u\rangle=e^{2\pi i\omega}|u\rangle$. In Table \ref{table:DPE}, we abbreviate "phase estimation algorithm" as PEA.  In addition, $B_u$ represents the number of qubits of $|u\rangle$. $n$ means that we want to estimate the first $n$ bits of the phase $\omega$, and $k$ is the number of computing nodes.
\end{remark}

It can be seen that our distributed phase estimation algorithm does not require quantum communication. Compared with the non-iterative phase estimation algorithm, the maximum number of qubits required by a single node of our distributed algorithm is reduced by $(1-\dfrac{1}{k})n-\log_2k-O(1)$. 

\subsection{Complexity of distributed order-finding algorithm}

	In Algorithm \ref{alg:DOFA}, similarly, we choose appropriate values for $l_2,\cdots,l_k$. We make $\left|l_i-(i-1)\cdot\dfrac{2L+2}{k}\right|<1$, $i=2,\cdots,k$. In addition, by utilizing the method mentioned in \cite{jiang2023distributed}, we can transmit $L$ qubits using only one EPR pair. Therefore, the quantum communications in Algorithm \ref{alg:DOFA} will only cause each node to add a maximum of 2 additional qubits. We directly present the following table to show the complexity of Algorithm \ref{alg:DOFA} and other related algorithms.

\begin{table}[width=1\linewidth,cols=4,pos=h]
\caption{Complexity of Algorithm \ref{alg:DOFA} and other related algorithms} \label{table:DOFA}
\begin{tabular*}{\tblwidth}{ p{2cm}| p{1.5cm}| p{1.5cm}| p{1.5cm}}
\bottomrule
Algorithms & qubits (per node) & quantum communication complexity& time complexity\\
\hline
Algorithm \ref{alg:DOFA} & $(2+\dfrac{2}{k})L+\log_2k+O(1)$ & $O(kL)$ & $O(L^3)$ \\\hline
traditional non-iterative OFA & $4L+O(1)$ & 0 & $O(L^3)$ \\\hline
Xiao's algorithm \cite{gang2023distributed} & $3.5L+O(1)$ & $O(L)$ & $O(L^3)$ \\\hline
Yimsiriwattana's distributed algorithm (\cite{yimsiriwattana2004distributed}, general method) & $L+O(1)$ & $O(L^2)$ & $O(L^3)$ \\
\toprule
\end{tabular*}
\end{table}


\begin{remark}
 In Table \ref{table:DPE}, we abbreviate "order-finding algorithm" as OFA. In addition, $L$ is the bit length of the number to be decomposed, and $k$ is the number of compute nodes.
\end{remark}

It can be seen that compared with the non-iterative order-finding algorithm, the maximum number of qubits required by a single node of our distributed algorithm is reduced by $(2-\dfrac{2}{k})L-\log_2k-O(1)$. 

In the following, we analyze the quantum communication complexity of Algorithm \ref{alg:DOFA}. First, Node $A_1$ transmits $L$ qubits from register $C$ to Node $A_2$. Then, Node $A_2$ transmits $L$ qubits from register $C$ to Node $A_3$, and so on.  Finally, Node $A_{k-1}$ transmits $L$ qubits from register $C$ to Node $A_k$. Thus, the quantum communication complexity of Algorithm \ref{alg:DOFA} is $O(kL)$.  So, the quantum communication complexity of our algorithm ($O(kL)$) is better than that of Yimsiriwattana's algorithm ($O(L^2)$). 

In addition, when we take $k=2$ (since Xiao's algorithm only uses $2$ nodes), the number of qubits required for each node of our algorithm is $3L+O(1)$, which is better than that of Xiao's algorithm.

\section{Conclusions} \label{sec:conclusions}

 In this paper, we have proposed a new distributed phase estimation algorithm. Our distributed phase estimation algorithm does not require quantum communication and it reduces the number of control qubits of a single node compared to the non-iterative phase estimation algorithm. It requires $\dfrac{n}{k}+\log_2+O(1)$ control qubits.  Afterwards, we have applied it to form a distributed order-finding algorithm for Shor's algorithm. Compared with the non-iterative Shor's algorithm, the maximum number of qubits required by a single node of our distributed order-finding algorithm is reduced by $(2-\dfrac{2}{k})L-\log_2k-O(1)$ when factoring an $L$-bit integer. It requires $(2+\dfrac{2}{k})L+\log_2k+O(1)$ qubits and its quantum communication complexity is $O(kL)$. 

However, we  have only studied the cases of  non-iterative phase estimation algorithm and non-iterative order-finding algorithm. In future research, we would consider to study the distributed algorithms for the iterative phase estimation algorithm and iterative order-finding algorithm. In addition, it is worthy of further consideration for applying our distributed phase estimation algorithm to HHL algorithm and discrete logarithm algorithm.
	

\section*{Appendix 1: Basic concepts of quantum computing}\label{appendix:quantum computing}
In quantum computing, quantum bits (qubit, for short) are basic units, and a qubit can be represented by a two-dimensional unit column vector, $\begin{bmatrix}\alpha\\ \beta\end{bmatrix}$ $\in \mathbb{C}^2$ (where $\mathbb{C}$ denotes the set of  complex numbers), $\begin{bmatrix}1\\ 0\end{bmatrix}$ and $\begin{bmatrix}0\\ 1\end{bmatrix}$ are two special quantum bits, called the computational basis states.
A general single-qubit state $\begin{bmatrix}\alpha\\ \beta\end{bmatrix}$ can be represented as $|\psi\rangle$ (Dirac notation), that is $|\psi\rangle$= $\begin{bmatrix}\alpha\\ \beta\end{bmatrix}$.
Thus, $|\psi\rangle$=$\alpha|0\rangle$+$\beta|1\rangle$, where $|\alpha|^2$+$|\beta|^2$=1.


A two-qubit system has four possible states: 00, 01, 10, 11. 
The computational basis states of a two-qubit system are $|00\rangle$, $|01\rangle$, $|10\rangle$, $|11\rangle$.
In general, $|ab\rangle$=$|a\rangle\otimes|b\rangle$, where $\otimes$ represents tensor product.
Therefore, any two-qubit state $|\psi\rangle$ can be represented as
$
	|\psi\rangle=\sum\limits_{i=0}^1\sum\limits_{j=0}^1\alpha_{i, j} |ij\rangle\text{,}
$
where $\alpha_{i, j}\in \mathbb{C}$  and
$\sum\limits_{i=0}^1\sum\limits_{j=0}^1\left|\alpha_{i, j}\right|^2=1$.


Furthermore, $n$-qubits can be represented as $|i_1i_2\cdots i_n\rangle$, where $i_j\in\{0,1\}$ and $j=1,2,\cdots,n$. 
Therefore, any $n$-$\rm qubit$ state $|\psi\rangle$ can be represented as
$
	|\psi\rangle=\sum\limits_{i_1,i_2,\cdots, i_n=0}^1\alpha_{i_1, i_2, \cdots, i_n} |i_1 i_2 \cdots i_n\rangle\text{,}
$
where $\alpha_{i_1, i_2, \cdots, i_n}\in \mathbb{C}$, and $\sum\limits_{i_1,i_2,\cdots,i_n=0}^1\left|\alpha_{i_1,i_2,\cdots,i_n}\right|^2=1$, $|i_1i_2\cdots i_n\rangle=|i_1\rangle\otimes|i_2\rangle\otimes\cdots\otimes|i_n\rangle$. 
$\langle\psi|$ represents the conjugate transpose of $|\psi\rangle$.

 Basic quantum gates include CNOT gate, $I$ gate, $Z$ gate,  $X$ gate, $Y$ gates,  Hadamard gate, and their definitions are as follows:
\begin{equation}
	\textrm{CNOT}=|00\rangle\langle 00|+| 01\rangle\langle 01|+|11\rangle\langle 10|+| 10\rangle\langle 11| \text{,}
\end{equation}
\begin{align}
	I&=|0\rangle\langle 0|+| 1\rangle\langle 1|\text{,}\\
	Z&=|0\rangle\langle 0|-| 1\rangle\langle 1|\text{,}\\
	X&=|1\rangle\langle 0|+| 0\rangle\langle 1|\text{,}\\
	Y&=-i(|0\rangle\langle 1|-| 1\rangle\langle 0|)\text{,}\\
	H&=\dfrac{1}{\sqrt{2}}(| 0\rangle+| 1\rangle)\langle 0|+\dfrac{1}{\sqrt{2}}(| 0\rangle-| 1\rangle)\langle 1|\text{.}
\end{align}

The state evolution of a closed quantum system is described by a unitary transformation, such that the transformation from any state $|\psi\rangle$ to state $|\psi'\rangle$  satisfies $|\psi'\rangle=U|\psi\rangle$, where $U$ is a unitary operator.
  Quantum measurement is described by a set of measurement operators $\{M_m\}$, which satisfy the completeness relation 
$\sum\limits_{m}M^\dag_mM_m=I$. These measurement operators act on the state of the system being measured, where $m$ is the measurement outcome. If the quantum state before measurement is $|\psi\rangle$, then the probability of obtaining $m$ is $p(m)=\langle \psi|M^\dag_mM_m|\psi\rangle$, and the state of the system collapses to
\begin{align}
\dfrac{M_m|\psi\rangle}{\sqrt{\langle \psi|M^\dag_mM_m|\psi\rangle}}\text{.}
\end{align}

\section*{Appendix 2: proof of correctness for distributed phase estimation algorithm}\label{sec:appendix}

\begin{lemma}\label{CorrectStep}
Let $s,x$ be two $t$-bit strings ($t\geq 3$). Let $y$ be a $3$-bit string. Suppose $d_t(x,s)\leq 1$ and $d_t(y,s_{[t-2,t]})\leq 1$. Then there only exist one element $b$ in $\{\pm 2,\pm 1,0\}$ such that $$ADD(x_{[t-2,t]},b)=y.$$
Moreover, let $b_1,b_2\in\{\pm1,0\}$ satisfy $ADD(x,b_1)=s$ and $ADD(s_{[t-2,t]},b_2)=y$, it holds that 
\begin{equation}
b=b_1+b_2.
\end{equation}
\end{lemma}
\begin{proof}
 Since $d_t(x,s)\leq 1$ , we have $d_t(x_{[t-2,t]},s_{[t-2,t]})\leq 1$. Then we get
\begin{equation}
d_t(x_{[t-2,t]},y)\leq d_t(x_{[t-2,t]},s_{[t-2,t]})+d_t(s_{[t-2,t]},y) \leq 2. 
\end{equation}
Hence, it is clear that such element $b$ is unique. Moreover, since
\begin{align}
ADD(x,b_1+b_2)&=ADD(ADD(x,b_1),b_2)\\
&=ADD(s,b_2)
\end{align}
we get
\begin{align}
ADD(x_{[t-2,t]},b_1+b_2)&=ADD(s_{[t-2,t]},b_2)\\
&=y
\end{align}
Therefore, $b=b_1+b_2$. The lemma holds.
\end{proof}

\begin{proposition}\label{prop:CorrectAndCombine}
Let $n,k,l_1,\cdots,l_{k+1}$ be postive integers with $1=l_1<\cdots<l_{k+1}=n-2$. Let $s$ be an $n$-bit string and let $x_i$ be an $(l_{i+1}+3-l_i)$-bit string such that 
$$d_{len(x_i)}(x_i,s_{[l_i,l_{i+1}+2]})\leq 1,$$ 
$i=1,\cdots,k$. Let $y=CorrectAndCombine(x_1,\cdots,x_k)$. Then 
$$d_n(y,s)=d_{len(x_k)}(x_k,s_{[l_k,l_{k+1}+2]})$$
holds.
\end{proposition}

\begin{proof}
Since $d_{len(x_i)}(x_i,s_{[l_i,l_{i+1}+2]})\leq 1$, we have 
\begin{equation}
d_{3}\left((x_i)_{[1,3]},s_{[l_i,l_i+2]}\right)\leq 1
\end{equation}
and
\begin{equation}
d_{3}\left((x_i)_{\left[len(x_i)-2,len(x_i)\right]},s_{[l_{i+1},l_{i+1}+2]}\right)\leq 1,
\end{equation}
$i=1,2,\cdots,k$. Let $b_1,b_2\in\{\pm1,0\}$ satisfy 
\begin{equation}
ADD\left((x_k)_{\left[len(x_k)-2,len(x_k)\right]},b_1\right)=s_{[l_{k+1},l_{k+1}+2]},
\end{equation}
\begin{equation}
ADD\left((x_k)_{[1,3]},b_2\right)=s_{[l_{k},l_{k}+2]},
\end{equation}
\begin{equation}
ADD\left((x_{k-1})_{\left[len(x_{k-1})-2,len(x_{k-1})\right]} ,b_3\right)=s_{[l_k,l_k+2]},
\end{equation}
Suppose we input $x_1,\cdots,x_k$ to Algorithm \ref{alg:CorrectAndCombine}. Let $prefix_r, m_r'$, $CorrectionNum_r$ ($r=1,\cdots,k-1$) be the same as those in Algorithm \ref{alg:CorrectAndCombine}. By Lemma \ref{CorrectStep}, we have $CorrectionNum_{k-1}=b_3-b_2$. Combining the Lemma 4 in \cite{gang2023distributed}, we get
\begin{align}
&\hspace{-2.5em}ADD(m_{k-1}',b_1)\\
&\hspace{-2.5em}=ADD(prefix_{k-1},b_2)\circ s_{[l_k+3,l_{k+1}+2]}\\
&\hspace{-2.5em}=ADD\left(ADD(x_{k-1},CorrectionNum_{k-1}),b_2\right)\\
&\circ s_{[l_k+3,l_{k+1}+2]}\\
&\hspace{-2.5em}=ADD(x_{k-1},b_3)\circ s_{[l_k+3,l_{k+1}+2]}\\
&\hspace{-2.5em}=s_{[l_{k-1},l_k+2]} \circ s_{[l_k+3,l_{k+1}+2]}\\
&\hspace{-2.5em}=s_{[l_{k-1},l_{k+1}+2]}.
\end{align}
Hence, 
\begin{align}
d_{len(m_{k-1}')}\left(m_{k-1}',s_{[l_{k-1},l_{k+1}+2]}\right)=|b_1|\\
=d_{len(m_{k}')}\left(m_{k}',s_{[l_{k},l_{k+1}+2]}\right).
\end{align}
By induction, it can be proven that
\begin{align}
\hspace{-2.5em}d_{len(m_{1}')}\left(m_{1}',s_{[1,l_{k+1}+2]}\right)=
d_{len(m_{k}')}\left(m_{k}',s_{[l_{k},l_{k+1}+2]}\right).
\end{align}
Since $y=m_1', m_k'=x_k$ and $l_{k+1}=n-2$, the proposition holds.
\end{proof}

	Finally, we prove Theorem \ref{CorrectnessForDPE}, which shows the correctness of Algorithm \ref{alg:DPE}.

\begin{proof}[\textbf{Proof of Theorem \ref{CorrectnessForDPE}}]
According to Corollary \ref{MeasurePE}, in Algorithm \ref{alg:DPE}, for any $i\in\{1,\cdots,k\}$, the probability of
\begin{equation}\label{eq:dt_mi_less1}
d_{len(m_i)}(m_i,\omega_{\{l_i,l_{i+1}+2\}})\leq 1
\end{equation}
is at least $1-\dfrac{\epsilon}{k}$. Hence, the probability of
Equation (\ref{eq:dt_mi_less1})
holds for all $i\in\{1,\cdots,k\}$ is at least $(1-\dfrac{\epsilon}{k})^k\geq 1-\epsilon$.
Finally by Proposition \ref{prop:CorrectAndCombine}, the theorem holds.
\end{proof}

\printcredits

\section*{Declaration of competing interest}

The authors declare that they have no known competing financial interests or personal relationships that could have appeared to influence the work reported in this paper.

\section*{Data availability}

No data was used for the research described in the article.

\end{document}